\documentclass[12pt, english]{article}
\usepackage{babel}
\usepackage{amsmath, amssymb, amsthm, eucal}
\usepackage{graphicx}
\usepackage[shortlabels]{enumitem}

\usepackage{color}

\usepackage[textheight=630pt,
  textwidth=468pt,
  centering]{geometry}

\newtheorem{theorem}{Theorem}
\newtheorem{corollary}{Corollary}

\theoremstyle{definition}
\newtheorem{definition}{Definition}

\newtheorem{remark}{Remark}

\newcommand{\eps}{\varepsilon}

\newcommand{\abs}[1]{\left| #1 \right|}

\newcommand{\norm}[1]{\Vert #1 \Vert}

\newcommand{\NN}{\ensuremath{\mathbb{N}}}

\newcommand{\RR}{\mathbb{R}}

\newcommand{\PP}{\mathbb{P}}

\newcommand{\Gauss}{\mathcal{N}}

\newcommand{\brac}[1]{\left(#1\right)}

\newcommand{\abrac}[1]{\langle#1\rangle}

\DeclareMathOperator{\mean}{\mathbb{E}}

\DeclareMathOperator{\Id}{Id}

\DeclareMathOperator{\circm}{circ}

\begin{document}


\title{Masked Toeplitz covariance estimation}
\author{Maryia Kabanava\thanks{RWTH Aachen University, Lehrstuhl C f{\"u}r Mathematik (Analysis), Pontdriesch 10, 52062 Aachen, Germany}{ } and Holger Rauhut\footnotemark[1]}

\maketitle

\begin{abstract}
The problem of estimating the covariance matrix $\Sigma$ of a $p$-variate distribution based on its $n$ observations arises in many data analysis contexts. While for $n>p$, the classical sample covariance matrix $\hat{\Sigma}_n$ is a good estimator for $\Sigma$, it fails in the high-dimensional setting when $n\ll p$. In this scenario one requires prior knowledge about the structure of the covariance matrix in order to construct reasonable estimators. 
Under the common assumption that $\Sigma$ is sparse, a refined estimator is given by $M\cdot\hat{\Sigma}_n$, where $M$ is a suitable symmetric mask matrix indicating the nonzero entries of $\Sigma$ and $\cdot$ denotes the entrywise product of matrices. 
In the present work we assume that $\Sigma$ has Toeplitz structure corresponding to stationary signals. This suggests to average the sample covariance $\hat{\Sigma}_n$ over the diagonals in order to obtain an estimator $\tilde{\Sigma}_n$ of Toeplitz structure. Assuming in addition that $\Sigma$ is sparse
suggests to study estimators of the form $M\cdot\tilde{\Sigma}_n$. For Gaussian random vectors and, more generally, random vectors satisfying the convex concentration property, our main result bounds the estimation error in terms of $n$ and $p$ and shows
that accurate estimation is indeed possible when $n \ll p$. The new bound significantly generalizes previous results by Cai, Ren and Zhou and provides an alternative proof. 
Our analysis exploits the connection between the spectral norm of a Toeplitz matrix and the supremum norm of the corresponding spectral density function.
\end{abstract}

\section{Introduction}



\subsection{Masked covariance estimation}

Estimating the covariance matrix of a random vector $X$ in $\RR^p$ from $n$ i.i.d.~sample observations $X_1,\hdots,X_n$ 
plays a key role in various data analysis tasks. Recently, the case $n \ll p$ of small sample size has attracted increasing attention due to its appearance in applications including mobile communication problems, gene expression studies and more. 


Let $X$ be random vector in $\RR^p$ which we assume to have mean zero throughout this article. 
(The general case of non-zero mean can be handled as in \cite[Remark 4]{LevinaVershynin}.) Its covariance matrix is defined as $\Sigma=\mean XX^T$. 
The sample covariance matrix of a sequence of $n$ i.i.d. observations $X_1,\hdots, X_n$ of $X$
is defined by
\begin{equation}\label{eq:SampleCovarianceMatrix}
\hat{\Sigma}_n=[\hat\sigma_{st}]_{s,t=1}^p=\frac{1}{n}\sum_{i=1}^n X_iX_i^T
\end{equation}
and it is an unbiased estimator of $\Sigma$. If $X$ is Gaussian and $n \geq C \varepsilon^{-2} p$ then the estimation error in the spectral norm satisfies 
\[
\|\hat{\Sigma}_n - \Sigma\| \leq \varepsilon
\]
 with probability at least $1-2 \exp(-cn)$, 
see e.g.~\cite[Corollarly 5.50]{Vershynin}, or \cite[Corollary 5.50]{Vershynin} 
for a variant for heavy-tailed distributions. Since the rank of $\hat{\Sigma}_n$ is at most $n$, the given bound of $n$ in terms of $p$ cannot be improved
for general $\Sigma$, i.e., a sample size of $n \geq p$ is necessary.
 
However, in modern applications it is desirable to find good estimators of the covariance matrix $\Sigma$ when $n\ll p$. Such estimators reflect prior knowledge about the structure of $\Sigma$. A common assumption is that $\Sigma$ is sparse, i.e., a significant amount of entries of $\Sigma$ is $0$ or close to $0$. Then the so-called masked covariance estimator is defined as $M\cdot\hat{\Sigma}_n$, where $M$ is a symmetric mask matrix and $\cdot$ denotes the entrywise product of matrices. Each entry $m_{ij}$ of $M$ indicates how important it is to estimate the interaction between the $i$-th and $j$-th variable. The masked approach was first introduced in \cite{LevinaVershynin} and it allows to describe several regularization techniques such as banding or tapering of the covariance matrix in the case of ordered variables \cite{BickelLevina,CaiRenZhou,FurrerBengtsson}, and thresholding in the case of unordered variables \cite{BickelLevinaThresholding,CaiZhou,ElKaroui}.

The accuracy of the masked estimator can be analyzed by splitting it into two terms via the triangle inequality 
\begin{equation}\label{eq:BiasVariance}
\norm{M\cdot\hat{\Sigma}_n-\Sigma}\leq\norm{M\cdot\hat{\Sigma}_n-M\cdot\Sigma}+\norm{M\cdot\Sigma-\Sigma},
\end{equation}
where $\norm{\cdot}$ denotes the spectral norm of a matrix. The bias term $\norm{M\cdot\Sigma-\Sigma}$ describes how well $\Sigma$ fits the model described by $M$. The variance term $\norm{M\cdot\hat{\Sigma}_n-M\cdot\Sigma}$ measures how accurately the part of the sample covariance matrix approximates the corresponding part of the true covariance matrix. The intuition behind the use of $M$ is that $M\cdot\Sigma$ preserves the essential structure of $\Sigma$, but at the same time $M\cdot\hat{\Sigma}_n$ does not deviate too much from its mean. 

In \cite{LevinaVershynin} the authors considered a $p$-variate Gaussian distribution and studied the problem of estimating the variance term $\norm{M\cdot\hat\Sigma_n-M\cdot\Sigma}$ for an arbitrary fixed symmetric $M\in\RR^{p\times p}$.
\begin{theorem}\label{th:LevinaVershyninResult}
Let $X_1,\ldots,X_n$ be drawn from a multivariate Gaussian distribution $\Gauss(0,\Sigma)$. Let $M\in\RR^{p\times p}$. Then
\begin{equation}\label{eq:SymmetricMaskGaussian}
\mean\norm{M\cdot\hat{\Sigma}_n-M\cdot\Sigma}\leq C\log^3(2p)\brac{\frac{\norm{M}_{1,2}}{\sqrt n}+\frac{\norm{M}}{n}}\norm{\Sigma},
\end{equation}
where $\norm{M}_{1,2}=\max\limits_{j}\brac{\sum\limits_i m_{ij}^2}^{1/2}$.
\end{theorem}
In the particular case when the entries of $M$ are either $0$ or $1$, estimate (\ref{eq:SymmetricMaskGaussian}) leads to the following corollary. 
\begin{corollary}
Let $X_1,\ldots,X_n$ be drawn from a multivariate Gaussian distribution $\Gauss(0,\Sigma)$. Assume that the entries of $M\in\RR^{p\times p}$ are equal to $0$ or $1$ and that there are at most $m$ nonzero entries in each column. Then
\begin{equation}\label{eq:SymmetricMaskZeroOneEntries}
\mean\norm{M\cdot\hat{\Sigma}_n-M\cdot\Sigma}\leq C\log^3(2p)\brac{\sqrt\frac{m}{n}+\frac{m}{n}}\norm{\Sigma}.
\end{equation}
\end{corollary}
The proof of Theorem \ref{th:LevinaVershyninResult} is based on decoupling, conditioning, covering argument and Gaussian concentration inequality for Lipschitz functions. It also allows to achieve error bounds that hold in probability. 

By means of a matrix moment inequality an error estimate that holds in expectation was generalized to arbitrary distributions with finite fourth moments in \cite{ChenGittensTropp}. When restricted to the Gaussian case it provides an improvement of (\ref{eq:SymmetricMaskGaussian}) in the logarithmic factor.


\subsection{Banding and tapering estimators of Toeplitz covariance matrices}

In this paper we are interested in obtaining bounds similar to (\ref{eq:SymmetricMaskGaussian}) and (\ref{eq:SymmetricMaskZeroOneEntries}) under the additional assumption that $X$ is stationary, resulting in the covariance matrix $\Sigma$ to be of Toeplitz structure, 
\begin{equation}\label{eq:ToeplitzMatrix}
\Sigma=\begin{pmatrix}
\sigma_0 & \sigma_1 & \sigma_2 & \dots & \ldots & \sigma_{p-1}\\
\sigma_1 & \sigma_0 & \sigma_1 & \ddots & & \vdots\\
\sigma_2 & \sigma_1 & \ddots & \ddots & \ddots & \vdots\\
\vdots & \ddots & \ddots & \ddots & \sigma_1 & \sigma_2\\
\vdots &  & \ddots & \sigma_1& \sigma_0 & \sigma_1\\
\sigma_{p-1} & \ldots & \ldots & \sigma_2 & \sigma_1 & \sigma_0
\end{pmatrix}.
\end{equation}
Stationary signals appear in many applications including time series analysis and mobile communications.
Our intuition is that the additional structure allows to further reduce the required number of samples.
The easiest way to improve the sample covariance estimator for this setting is to average the entries of $\hat{\Sigma}_n$ in (\ref{eq:SampleCovarianceMatrix}) over the diagonals. For $0\leq r\leq p-1$, set
\[
\tilde\sigma_r=\frac{1}{p-r}\sum_{s-t=r}\hat\sigma_{st}
\] 
and define a new unbiased estimator as the Toeplitz matrix $\tilde{\Sigma}_n=[\tilde\sigma_{st}]_{s,t=1}^p$ with $\tilde\sigma_{st}=\tilde\sigma_{\abs{s-t}}$. 

Assuming that there is an ordering among the variables of $X$ and that the variables, which are far apart, 
are only weakly correlated, we may construct more accurate estimators of $\Sigma$, so called banding and tapering
estimators \cite{CaiRenZhou}, that we describe here with the mask formalism. 
For a given positive integer $m\leq\frac{p}{2}$ and $0\leq r\leq p-1$, set 
\[
a_r=\left\{\begin{array}{ll}
1, & r\leq\frac{m}{2},\\
2-\frac{2r}{m}, &  \frac{m}{2}<r\leq m,\\
0, &  \text{otherwise},
\end{array}\right.\quad\text{and}\quad b_r=\left\{\begin{array}{ll}
1, & r\leq m,\\
0, & \text{otherwise}.
\end{array}\right.
\]
The tapering and banding estimators are defined as $M_{\operatorname{tap}}\cdot\tilde{\Sigma}_n$ and $M_{\operatorname{band}}\cdot\tilde\Sigma_n$, where 
\[
(M_{\operatorname{tap}})_{st}=a_{\abs{s-t}}\quad \text{ and } \quad (M_{\operatorname{band}})_{st}=b_{\abs{s-t}}.
\]
A bound on the variance term in the error bound (\ref{eq:BiasVariance}) for either the tapering or banding mask can be derived from results in \cite{CaiRenZhou}, see eq.\ (22) and Lemma 5 of loc. cit. 
\begin{theorem}\label{thm:crz}
Let $X_1,\ldots,X_n$ be drawn from a multivariate Gaussian distribution $\Gauss(0,\Sigma)$. Let $M\in\RR^{p\times p}$ be a tapering or banding mask with $m\leq\frac{p}{2}$. Then
\begin{equation}\label{eq:BandingTaperingEstimators}
\mean\norm{M\cdot\tilde\Sigma_n-M\cdot\Sigma}\leq C\sqrt\frac{m\log(np)}{np}.
\end{equation}
\end{theorem}

Comparing (\ref{eq:SymmetricMaskZeroOneEntries}) and (\ref{eq:BandingTaperingEstimators}), we see that there is an improvement by a factor of $\frac{1}{\sqrt p}$ in the error bound for Toeplitz matrices. However, the result (\ref{eq:BandingTaperingEstimators}) holds only for the special type of masks $M$, whereas (\ref{eq:SymmetricMaskZeroOneEntries}) is valid for any symmetric $M\in\RR^{p\times p}$. Our goal is to extend the error estimate (\ref{eq:BandingTaperingEstimators}) to general Toeplitz masks and not necessarily Gaussian distributions. 


\subsection{Our contribution}

Our result holds for distributions that satisfy the so-called convex concentration property, see Definition \ref{def:CCP} below which includes mean-zero Gaussian random vectors. The class of such distribution is, however, much broader than the Gaussian class. 
For  a Toeplitz mask $M\in\RR^{p\times p}$ with $M_{st}=\omega_{\abs{s-t}} \geq 0$, we define the weighted $\ell_1$- and $\ell_2$-norm of its first row $\omega=\brac{\omega_{\ell}}_{\ell=0}^{p-1}$ by
\[
 \norm{\omega}_{1,*}=\sum_{\ell=0}^{p-1}\frac{\omega_{\ell}}{p-\ell}\quad\text{and}\quad \norm{\omega}_{2,*}= \brac{\sum_{\ell=0}^{p-1}\frac{\omega_{\ell}^2}{p-\ell}}^{1/2}.
\]


\begin{theorem}\label{th:MainResult}
Let $X_1,\ldots,X_n$ be drawn from a distribution $X\in\RR^p$ such that $\mean X=0$, $\Sigma=\mean XX^T$ is Toeplitz and $X$ satisfies the convex concentration property with constant $K$. Let the Toeplitz mask $M\in\RR^{p\times p}$ with first row $\omega$. Then for every $t>0$,
\begin{equation}\label{prob:bound}
\PP\brac{\norm{M\cdot\tilde\Sigma_n-M\cdot\Sigma}\geq CK^2\brac{\norm{\omega}_{2,*}\sqrt\frac{t}{n}+\frac{\norm{\omega}_{1,*}t}{n}}}\leq Cpe^{-t},
\end{equation}
and
\begin{equation}\label{mean:bound}
\mean\norm{M\cdot\tilde\Sigma_n-M\cdot\Sigma}\leq CK^2\brac{\norm{\omega}_{2,*}\sqrt{\frac{\log(p)}{n}}+\norm{\omega}_{1,*}\frac{\log(p)}{n}}.
\end{equation}
\end{theorem}
We note that error bound \eqref{mean:bound} in expectation also holds in the mean square error (MSE), as follows easily from
the probability bound \eqref{prob:bound} together with integration. The logarithmic factor in \eqref{mean:bound} cannot be 
removed in general, see also Section~\ref{sec:smooth} below.

By estimating the weighted $\ell_1$ and $\ell_2$-norm of either the tapering or banding mask, we obtain the following result which generalizes and very slightly improves
Theorem~\ref{thm:crz} above from \cite{CaiRenZhou}. Moreover, it provides an alternative proof.

\begin{corollary}\label{cor:Gaussian}
Let $X_1,\ldots,X_n$ be drawn from a distribution $X\in\RR^p$ such that $\mean X=0$, $\Sigma=\mean XX^T$ is Toeplitz and $X$ satisfies the convex concentration property with constant $K$.
Let $M\in\RR^{p\times p}$ be a tapering or banding mask with $m\leq\frac{p}{2}$. Then 
\begin{equation}
\mean\norm{M\cdot\tilde\Sigma_n-M\cdot\Sigma}\leq C K^2 \brac{\hspace{-0.5mm}\sqrt{\frac{m\log(p)}{pn}}+\frac{m\log(p)}{pn}}.
\end{equation}
\end{corollary}
In the Gaussian case $X \sim {\mathcal{N}}(0, \Sigma)$, we have $K^2= 2 \norm{\Sigma}$, so that 
\begin{equation}
\mean\norm{M\cdot\tilde\Sigma_n-M\cdot\Sigma}\leq C\norm{\Sigma}\brac{\hspace{-0.5mm}\sqrt{\frac{m\log(p)}{pn}}+\frac{m\log(p)}{pn}}.
\end{equation}
Corollary \ref{cor:Gaussian} implies that for an error tolerance $\eps\in(0,1)$, the sample size
\[
n\geq C^2\eps^{-2}\frac{m}{p}\log(p)
\]
is sufficient for
\[
\mean\norm{M\cdot\tilde\Sigma_n-M\cdot\Sigma}\leq \eps\norm{\Sigma}.
\] 
Therefore, even though the number of observations may be significantly smaller than the dimension of the underlying distribution, partial estimation of the covariance matrix is performed with small error.  

In some applications, Toeplitz covariance matrices may have a sparsity structure that is more complicated than the one induced by the banding estimator
in the sense that zeros in $\omega$ interleave with nonzero entries. This includes spectrum sensing applications \cite{ZengLiang} 
where one needs to test the occupancy of spectral bands for wireless communication purposes. Man made signals, for instance in OFDM \cite{chaudhari}, may have statistics with Toeplitz covariance matrices and 
non-trivial sparsity structure with some zeros close to the diagonal and some non-zeros far away from the diagonal. 

Denoting by $S \subset \{0,\hdots,p-1\}$ the support of the first row of a sparse Toeplitz matrix $\Sigma$, it is natural to work with a Toeplitz mask $M$ having
first row $\omega = \mathbf{1}_S$ being the indicator of $S$, i.e., $\omega_j = 1$ for $j \in S$ and $\omega_j = 0$ for $j \notin S$.
The following weighted version of the cardinality of $S$, introduced in similar form in \cite{rauhutward}, determines the required number of samples,
\[
\nu(S) = \sum_{\ell \in S} \frac{p}{p-\ell}.
\]
Observe that for $\omega = \mathbf{1}_S$, $\nu(S) = p \|\omega\|_{1,*} = p \|\omega\|_{2,*}^2$. If $S$ is contained in a band of length $q$, i.e., 
$S \subset \{0,\hdots,q\}$ then
$\nu(S) \leq \frac{p}{p-q} \#S$ and for $q \leq p/2$ we have $\#S \leq \nu(S) \leq 2 \#S$. The following is an immediate consequence of Theorem~\ref{th:MainResult}. 

\begin{corollary}\label{cor:SparseMask}
Let $X_1,\ldots,X_n$ be drawn from a distribution $X\in\RR^p$ such that $\mean X=0$, $\Sigma=\mean XX^T$ is Toeplitz and $X$ satisfies the convex concentration property with constant $K$. Let $M\in\RR^{p\times p}$ be sparse Toeplitz with $M_{st}=\omega_{\abs{s-t}}\in\{0,1\}$ where $\omega$ has
support $S \subset \{0,\hdots,p-1\}$ of weighted cardinality $\nu(S)$.
Then for every $t>0$,
\[
\PP\brac{\norm{M\cdot\tilde\Sigma_n-M\cdot\Sigma}\geq CK^2\brac{\sqrt\frac{\nu(S) t}{pn}+\frac{\nu(S) t}{pn}}}\leq Cpe^{-t},
\]
and
\[
\mean\norm{M\cdot\tilde\Sigma_n-M\cdot\Sigma}\leq CK^2\brac{\sqrt{\frac{\nu(S) \log(p)}{pn}}+\frac{\nu(S) \log(p)}{pn}}.
\]
\end{corollary}

In the case of a Gaussian distribution this result reduces to the previously known Theorem~\ref{thm:crz} for the banding estimator if $S = \{0,\hdots, m\}$, but may handle general sparsity patterns (and more general distributions).
For instance, if $S \subset \{0,\hdots,p/2\}$ of (small) cardinality $s = \#S$, then as few as
\[
n \geq C \varepsilon^{-2} \frac{s}{p} \log(p)
\] 
samples ensure $\norm{M\cdot\tilde\Sigma_n-M\cdot\Sigma}\leq \eps\norm{\Sigma}$ with high probability.

\begin{remark}
It would be interesting to investigate whether the error bound \eqref{mean:bound} in expectation can be generalized
to heavier tailed distributions in the spirit of the main results in \cite{ChenGittensTropp} which only assume finite fourth moments. It is however presently not clear whether it is possible to adapt the proof technique of \cite{ChenGittensTropp} to our Toeplitz covariance structure.
\end{remark}

\subsection{Bounds over a class of smooth spectral densities}
\label{sec:smooth}

Let us shortly describe an application of our results studied in more detail in \cite{CaiRenZhou}.
To a Toeplitz covariance matrix $\Sigma$ of the form \eqref{eq:ToeplitzMatrix} we associate its spectral density function
\begin{equation}\label{eq:SpectralDensityFunction}
f(x) = f_\Sigma(x) =\sigma_0+2\sum_{r=1}^{p-1}\sigma_r\cos rx,
\quad x\in[-\pi,\pi]. 
\end{equation}
The proof of our main result uses the fact that the spectral norm of $\Sigma$ can be estimated by the $L^\infty$ norm
of $f$, $\norm{\Sigma}\leq \|f\|_\infty := \sup\limits_{x\in[-\pi,\pi]}\abs{f(x)}$, see e.g.\ \cite[Chapter 5.2]{GrenanderSzegoe}.

As in \cite{CaiRenZhou} we introduce a class of Toeplitz covariance matrices related to a Lipschitz condition on the spectral densities. 
For $\beta = \gamma + \alpha$ with $\gamma \in \NN_0$ and $\alpha \in (0,1]$, let 
\[
\mathcal{F}_\beta(L_0, L) = \left\{\Sigma \succ 0 : \|\Sigma\| \leq L_0, \sup_{x \in [-\pi,\pi]} |f_\Sigma^{(\gamma)}(x+h) - f_\Sigma^{(\gamma)}(x)| \leq L h^{\alpha} \right\},
\]
where $f_\Sigma^{(\gamma)}$ is the $\gamma$-th derivative of $f_\Sigma$. Since the decay of Fourier coefficients is closely connected
to smoothness conditions, these two classes are contained in each other for certain choices of parameters.

Choosing $M = M_{\operatorname{tap}}$ as the mask of the tapering estimator with parameter $m$, 
the spectral function $f_{M_{\operatorname{tap}}\cdot \Sigma}$ equals $f_\Sigma * V_m$ where $*$ denotes convolution and 
$V_m$ is the so-called De la Vall{\`e}e-Poussin kernel. Applying classical results from Fourier series, see e.g.\ 
\cite[Chapter 3.13]{Zygmund} yields $\|M_{\operatorname{tap}} \cdot \Sigma - \Sigma\| \leq \| f_\Sigma * V_m - f_\Sigma\|_\infty 
\leq 4 \inf_{q \in T_m} \|q-f\|_\infty$. For $\Sigma$ in $\mathcal{F}_\beta(L_0,L)$ the last term can further be estimated
by $3L m^{-\beta}$ so that 
\begin{equation}\label{est:mtap}
\|M_{\operatorname{tap}} \cdot \Sigma - \Sigma\| \leq 12 L m^{-\beta}.
\end{equation}
Under the assumption of 
Corollary~\ref{cor:Gaussian} and assuming $K^2 = c \norm{\Sigma}$ (as in the Gaussian case) this leads together with \eqref{eq:BiasVariance} to 
\[
\mean \|M_{\operatorname{tap}}\cdot \tilde{\Sigma} - \Sigma\| \leq C \norm{\Sigma} \left(\sqrt{\frac{m \log(p)}{pn}} + \frac{m \log(p)}{pn}\right) + 12 L m^{-\beta}
\]
Note that $\norm{\Sigma} \leq L_0$ due to $\Sigma \in {\mathcal{F}}(L_0,L)$. 
Choosing 
\begin{equation}\label{mtap:choice}
m = \left\lfloor \left(\frac{L}{L_0^{2\beta + 2}} \frac{np }{\log(p)}\right)^{1/(2\beta+1)}  \right\rfloor
\end{equation}
and making the mild assumption $m \leq pn/\log(p)$,
we obtain
\begin{equation}\label{tap:estimate}
\mean \|M_{\operatorname{tap}} \cdot \tilde{\Sigma} - \Sigma\| \leq C L \left( \frac{\log(p)}{np} \right)^{\frac{\beta}{2\beta+1}}.
\end{equation}
Of course, a related probability estimate and an MSE estimate can be derived from \eqref{prob:bound}.
It is shown in \cite[Theorem 5]{CaiRenZhou} that this bound is optimal over the class $\mathcal{F}_\beta(L_0,L)$.
In particular, the logarithmic factor $\log(p)$ in \eqref{tap:estimate} cannot be removed. This means that the logarithmic
factor in our general bound \eqref{mean:bound} cannot be removed in general, either. 

In a similar way \cite{CaiRenZhou}, we can analyze the performance of the banding estimator over $\mathcal{F}_\beta(L_0,L)$. 
This leads to the estimate
\[
\mean \|M_{\operatorname{band}} \cdot \tilde{\Sigma} - \Sigma\| \leq C L_0 \left(\sqrt{\frac{m \log(p)}{pn}} + \frac{m \log(p)}{pn}\right) + 12 L \log(m) m^{-\beta}.
\]
Choosing 
\[
m = \left\lfloor \left(\frac{L}{L_0^{2\beta + 2}} \frac{np }{\log(p)}\right)^{1/(2\beta+1)} \log(p)^{1/\beta}  \right\rfloor
\]
leads to
\[
\mean \|M_{\operatorname{band}} \cdot \tilde{\Sigma} - \Sigma\| \leq C L \left( \frac{\log(p)^{\frac{4\beta+3}{4\beta+2}}}{np} \right)^{\frac{\beta}{2\beta+1}}.
\]
Compared to the bound for the tapering estimator this is slightly worse.

The article \cite{CaiRenZhou} considers also a second class of Toeplitz covariance matrices, but for the sake of brevity, we will not go into detail here.


\subsection{Positive semidefinite estimator}

It is natural to ask that a covariance estimator is positive semidefinite and some applications will strictly require this.
However, our masked estimator $M\cdot\tilde\Sigma_n$ does not necessarily fulfill this condition. In order to obtain a positive semidefinite Toeplitz estimator, we can apply a procedure described in \cite[Chapter 5]{CaiRenZhou}, which is based on a circulant extension of our original masked estimator. For the sake of completeness we present the construction here and provide an approximation error of the true covariance matrix. 

For an arbitrary (positive semidefinite) Toeplitz covariance matrix $\Sigma\in\RR^{p\times p}$ given by (\ref{eq:ToeplitzMatrix}) and corresponding spectral density function $f_{\Sigma}$ as in (\ref{eq:SpectralDensityFunction}) we define a circulant matrix $\Sigma_{\circm}\in\RR^{(2p-1)\times(2p-1)}$ with entries
\[
\brac{\Sigma_{\circm}}_{st}=\begin{cases}
\sigma_{\abs{s-t}}, & \text{if } \abs{s-t}\leq p-1,\\
\sigma_{2p-1}-\sigma_{\abs{s-t}}, & \text{if } p\leq\abs{s-t}\leq 2p-2.
\end{cases}
\]
The eigenvalue decomposition of $\Sigma_{\circm}$ is given by
\[
\Sigma_{\circm} =\sum_{\abs j\leq p-1}\lambda_ju_j\bar{u}_j^T
\]
with the eigenvectors
\[
u_j=\frac{1}{\sqrt{2p-1}}\brac{1,e^{-2\pi i j/(2p-1)},\ldots,e^{-2\pi i j(2p-2)/(2p-1)}}^T,\quad \abs{j}\leq p-1,
\]
and the (non-negative) eigenvalues 
\begin{align}
\lambda_j&= \sum_{r=0}^{p-1}\sigma_re^{-2\pi i jr/(2p-1)} +  \sum_{r=p}^{2p-2}\sigma_{2p-1-r} e^{-2\pi i jr/(2p-1)} \notag\\
&= \sum_{r=0}^{p-1}\sigma_re^{-2\pi i jr/(2p-1)} +  \sum_{r=-(p-1)}^{-1}\sigma_{-r} e^{-2\pi i jr/(2p-1)}\notag\\
&=\sigma_0+2\sum_{r=1}^{p-1}\sigma_r\cos \frac{2\pi rj}{2p-1}=f_{\Sigma}\brac{\frac{2\pi j}{2p-1}},\quad \abs{j}\leq p-1\label{eq:EigenValuesOfCirculantMatrix}.
\end{align}

Let $M\cdot\tilde{\Sigma}_n$ be our masked estimator with spectral density function $f_{M\cdot\tilde{\Sigma}_n}$, which may possibly take negative values. Define $f^*:[-\pi,\pi]\to\RR$ as the non-negative part of $f_{M\cdot\tilde\Sigma_n}$, 
 \[
 f^*(x)=\begin{cases}
 f_{M\cdot\tilde{\Sigma}_n}(x), & \text{if } f_{M\cdot\tilde{\Sigma}_n}(x)\geq 0,\\
 0, & \text{otherwise}.
 \end{cases}
 \]
Set
\[
\Sigma^*_{\circm} = \sum_{\abs{j}\leq p-1}f^*\brac{\frac{2\pi j}{2p-1}}u_j\bar{u}_j^T.
\]
Then $\Sigma^*_{\circm}\in\RR^{(2p-1)\times(2p-1)}$ is circulant and positive semidefinite. As a new estimator $\Sigma^*$ we take the restriction of $\Sigma^*_{\circm}$ to its first $p$ rows and $p$ columns. It is clear that $\Sigma^*$ is Toeplitz and positive semidefinite. But note that in the case of a sparse mask $M$, the estimator $\Sigma^*$ may in general fail to be sparse. Nevertheless, we have the following error bound.
\begin{theorem}\label{th:Perfomance OfPositiveEstimator}  
Let $X_1,\ldots,X_n$ be drawn from a distribution $X\in\RR^p$ such that $\mean X=0$, $\Sigma=\mean XX^T$ is Toeplitz and $X$ satisfies the convex concentration property with constant $K$. Let $M\in\RR^{p\times p}$ be a Toeplitz mask  with first row $\omega$. Then the positive semidefinite estimator $\Sigma^*$ obtained from $M\cdot\tilde\Sigma_n$ by the procedure described above satisfies, for every $t>0$,
\begin{equation}\label{prob:pos-def}
\PP\left(\norm{\Sigma^*-\Sigma}\geq CK^2\brac{\norm{\omega}_{2,*}\sqrt{\frac{t}{n}}+\frac{\norm{\omega}_{1,*} t}{n}}+3\norm{f-f_{M\cdot \Sigma}}_\infty\right) \leq C p e^{-t},
\end{equation}
where $f$ and $f_{M\cdot \Sigma}$ denote spectral density functions of $\Sigma$ and $M\cdot\Sigma$ respectively. Moreover,
\begin{equation}\label{error:pos-def}
\mean\norm{\Sigma^*-\Sigma}\leq  CK^2\brac{\norm{\omega}_{2,*}\sqrt{\frac{\log(p)}{n}}+\norm{\omega}_{1,*}\frac{\log(p)}{n}}+3\norm{f-f_{M\cdot\Sigma}}_\infty,
\end{equation}
\end{theorem}
The term $\norm{f-f_{M\cdot \Sigma}}_\infty$ in \eqref{error:pos-def} replaces the bias term in \eqref{eq:BiasVariance} and is in general an upper bound for it. If the mask $M$ is chosen in such a way that $M\cdot\Sigma$ is precisely $\Sigma$ then the term $\norm{f-f_{M\cdot\Sigma}}_\infty$ in the estimate above disappears.  

In the situation of Toeplitz covariance matrices $\Sigma$ from the class $\mathcal{F}_\beta(L_0,L)$ from Section~\ref{sec:smooth} we have for the tapering mask $M_{\operatorname{tap}}$ with parameter $m$, see also \eqref{est:mtap},
\[
\norm{f-f_{M_{\operatorname{tap}} \cdot\Sigma}}_\infty \leq 12 L m^{-\beta}.
\]
Choosing $m$ as in \eqref{mtap:choice} and following the same steps leading to \eqref{tap:estimate}, 
we conclude that the corresponding positive definite estimator $\Sigma^*$ satisfies
\[
\mean \| \Sigma^* - \Sigma\| \leq CL \left( \frac{\log(p)}{np} \right)^{\frac{\beta}{2\beta+1}}.
\]
A corresponding tail estimate follows in the same way.
This means that the original masked estimator $M \cdot \tilde\Sigma_n$ and the positive semidefinite estimator 
$\Sigma^*$ obey the same error estimates on $\mathcal{F}_\beta(L_0,L)$ (up to possibly constants). 

\section*{Acknowledgements}

Both authors acknowledge funding from the DFG through the project Compressive Covariance Sampling for Spectrum Sensing (CoCoSa). They thank Andreas Bollig and Arash Behboodi for discussions on Toeplitz covariance estimation in the context of wireless communications.
 
\section{Preliminaries}
\subsection{The convex concentration property}

The Gaussian concentration inequality, see e.g.~\cite{Ledoux,LedouxTalagrand}, states that if $f:\RR^p\to\RR$ is a Lipschitz function with Lipschitz
constant  $\norm{f}_\text{Lip}$ and $X\sim\Gauss(0,\Sigma)$, then
\begin{equation}\label{eq:GaussianConcentration}
\PP\brac{\abs{f(X)-\mean f(X)}\geq t}\leq 2\exp\brac{-\frac{t^2}{2\norm{\Sigma}\norm{f}_{\text{Lip}}^2}}\quad\text{for all } t\geq 0.
\end{equation}
We are interested in distributions $X\in\RR^p$ that behave similar to (\ref{eq:GaussianConcentration}). 
\begin{definition}[Convex concentration property]\label{def:CCP}
Let $X$ be a random vector in $\RR^p$. We say that $X$ has the convex concentration property (c.c.p.) with constant $K$ if for every $1$-Lipschitz convex function $\phi:\RR^p\to\RR$, we have $\mean\abs{\phi(X)}<\infty$ and for every $t>0$,
\[
\PP(\abs{\phi(X)-\mean\phi(X)}\geq t)\leq 2\exp(-t^2/K^2).
\]
\end{definition}
The mean $\mean \phi(X)$ in \eqref{eq:GaussianConcentration} may be replaced by a median $M_f$ after possibly adjusting the constant $K$, see e.g.~\cite[Lemma 3.2]{Adamczak}. 

This type of distributions is considered in \cite{AdamczakLogSobolev,Adamczak,Ledoux,VuWang,MeckesSzarek}. We provide several examples of distributions with possibly dependent entries satisfying the c.c.p.:
\begin{enumerate}
\item Clearly, a Gaussian random vector $X\sim\Gauss(0,\Sigma)$ satisfies the c.c.p. with $K^2=2\norm{\Sigma}$.
\item A random vector $X$ that is uniformly distributed on the sphere $\sqrt{p}S^{p-1}$ satisfies the c.c.p.\ with constant $K=2$. This follows from \cite[eq.\ (5.7)]{Ledoux} combined with \cite[Theorem 5.3]{Ledoux}.
\item A random vector $X\in\RR^p$ with a density proportional to $e^{-u(x)}$, where the Hessian satisfies $D^2 u(x)\geq \gamma \Id$ for some $c>0$ uniformly in $x\in\RR^p$, has the c.c.p.~\cite[Proposition 2.18]{Ledoux} with constant $K= \sqrt{2/\gamma}$. Such random vectors form an important subclass of the logarithmically convex random vectors.
\item\label{it:IndependentEntries} A random vector $X = (X_1,\hdots,X_p)$ with independent components $X_j$ taking values in $[-1,1]$ satisfies the c.c.p.\ with absolute constant $K= c$ \cite{Talagrand:1995tn}. (Of course, the $X_j$ taking values in some other bounded intervals works as well after possibly adjusting the constant $c$.)
\item In generalization of the previous example, the c.c.p.\ also holds for certain random vectors $X = (X_1,\hdots,X_p)$ on $[-1,1]^p$ with dependent entries. In \cite{Samson:2000bf} this is proven for some classes of Markov chains and so-called $\Phi$-mixing processes. 

\item Let $X\in\RR^p$ be a random vector with covariance matrix being the identity and 
that satisfies the c.c.p.\ with constant $K$, for instance, a Rademacher vector, i.e., independent entries that take the value $\pm 1$ with equal probability (see Example~\ref{it:IndependentEntries}). Now for an arbitrary 
$B\in\RR^{q\times p}$ we define $Y = B X \in \RR^q$. Then $Y$ has covariance matrix $\Sigma_Y=BB^T$ and
satisfies the c.c.p. Indeed, for a $1$-Lipschitz and convex function $f:\RR^q\to\RR$, define 
$\phi:\RR^p\to\RR$ as $\phi(X)=\frac{1}{\norm{B}}f(BX)$. Then $\phi$ is also $1$-Lipschitz and convex.
Since $X$ has the c.c.p., we have 
\[
\mean\abs{f(Y)}=\mean\abs{f(BX)}=\norm{B}\mean\abs{\phi(X)}<\infty
\]
and 
\[
\PP(\abs{f(Y)-\mean f(Y)}\geq t)=\PP\brac{\abs{\phi(X)-\mean \phi(X)}\geq \frac{t}{\norm{B}}}\leq 2\exp(-t^2/(K\norm{B})^2),
\]
which implies that $Y$ satisfies the c.c.p. with constant $K\norm{\Sigma_Y}^{1/2}$.
This example shows in particular that any positive semidefinite matrix $\Sigma$ may appear as covariance matrix of
a random vector satisfying the c.c.p. and not being Gaussian.
\item It follows from \cite[Theorem 3.3]{Paulin:2014fj} (a generalization of Talagrand's convex distance inequality) that the c.c.p.\ holds for a random vector $X = (X_1,\hdots,X_p)$ with possibly dependent entries
which satisfies a Dobrushin type condition. See \cite{Paulin:2014fj} for details and \cite[Theorem 3]{Talagrand:1988ki} on how to deduce a concentration inequality from a convex distance inequality. 
This examples applies in particular to random vectors generated via sampling from finite sets without replacement \cite[Theorem 6.8]{Paulin:2014fj}. 
\item Random vectors satisfying the logarithmic Sobolev inequality are c.c.p.: For some positive measurable function $f$ on $\RR^p$, the entropy is defined as
\[
\operatorname{Ent}_{X}(f) = \mean[f(X) \log(f(X))] - \mean[f(X)] \log(\mean[f(X)]).
\]
The random vector is said to satisfy a logarithmic Sobolev inequality if for all smooth enough functions $f$ on $\RR^p$ it holds
\[
\operatorname{Ent}_{X}(f^2) \leq K^2 \mean\left[ \| \nabla f(X) \|_2^2\right]. 
\]
It follows from \cite[Theorem 5.3]{Ledoux} that $X$ has the c.c.p.\ with constant $K$.
Examples of random vectors satisfying the logarithmic Sobolev inequality include 
Gaussian random vectors $X\sim\Gauss(0,\Sigma)$ and more generally logarithmically concave random vectors as in Example 3.\ above, the uniform distribution on the sphere \cite[eq.\ (5.7)]{Ledoux} and, more generally,
random vectors distributed according to the normalized Riemann measure on a compact Riemannian manifold with Ricci curvature uniformly bounded from below by a positive constant \cite[eq.\ (5.5)]{Ledoux}.

\end{enumerate}
The following generalization of the Hanson-Wright inequality for random vectors satisfying the c.c.p.\ due to Adamczak \cite{Adamczak} is crucial for the proof of Theorem \ref{th:MainResult}.  
\begin{theorem}\label{th:HansonWright}
Let $X\in\RR^p$ be random with $\mean X=0$. If $X$ satisfies the c.c.p.\ with constant $K$, then for any $A\in\RR^{p\times p}$ and $t>0$,
\[
\PP\brac{|\abrac{AX,X}-\mean\abrac{AX,X}|\geq t}\leq 2\exp\brac{-\frac{1}{C}\min\brac{\frac{t^2}{2K^4\norm{A}_F^2},\frac{t}{K^2\norm{A}}}}.
\]
\end{theorem}

\subsection{Sub-gamma random variables}

The proof of our main result uses the concept of sub-gamma random variables, see also \cite[Chapter 2.4]{BoucheronLugosiMassart}.
A real-valued mean-zero random variable $X$ is called sub-gamma with variance factor $\nu$ and scale parameter $c$ if, for all $0<\lambda<1/c$,
\begin{equation}\label{def:subgamma}
\mean \exp(\lambda X)  \leq \exp\left(\frac{\lambda^2 \nu}{2(1-c\lambda)}\right) \quad \mbox{ and } \quad 
\mean \exp(-\lambda X) \leq \exp\left( \frac{\lambda^2 \nu}{2(1-c\lambda)}\right) 
\end{equation}
The tail of a sub-gamma variable satisfies \cite[Chapter 2.4]{BoucheronLugosiMassart}
\begin{equation}\label{subgamma:tail}
\PP(|X| > \sqrt{2 \nu t} + \sqrt{c t}) \leq 2 \exp(-t). 
\end{equation}
Sub-gamma variables can be characterized via their moments \cite[Theorem 2.3]{BoucheronLugosiMassart}.
\begin{theorem}\label{thm:subgamma} 
If, for any integer $q \geq 1$, a random variable $X$ satisfies
\begin{equation}\label{char:subgamma}
\mean[X^{2q}] \leq q! A^q + (2q)! B^{2q}
\end{equation}
then $X$ is sub-gamma with variance factor $\nu = 4(A+B^2)$ and scale parameter $c= 2B$. 
\end{theorem}
Conversely, if $X$ is sub-gamma then \eqref{char:subgamma} holds for some $A$ and $B$.  

\section{Proof of main results}

\begin{proof}[Proof of Theorem \ref{th:MainResult}]
As in \cite{CaiRenZhou} our results rely on the connection between the spectral norm of the Toeplitz matrix and the $L^{\infty}$ norm of the corresponding spectral density function.

The spectral density function corresponding to a Toeplitz covariance matrix $\Sigma$ defined in \eqref{eq:ToeplitzMatrix} is given by
\[
f(x)=f_{\Sigma}(x)=\sigma_0+2\sum_{r=1}^{p-1}\sigma_r\cos rx=\sum\limits_{r=-(p-1)}^{p-1}\sigma_{\abs{r}}\cos rx,\quad x\in[-\pi,\pi]. 
\]
It follows from \cite[Chapter 5.2]{GrenanderSzegoe} that 
\begin{equation}\label{eq:SpectralNormMaxSpectralDensity}
\norm{\Sigma}\leq\norm{f}_{\infty}:=\sup\limits_{x\in[-\pi,\pi]}\abs{f(x)} \leq 2\max\limits_{1\leq k\leq 4p}\abs{f(x_k)} , \quad \mbox{ with } x_k = \frac{k-2p}{4p} \pi,
\end{equation}
where the second inequality follows from the fact that $f$ is a trigonometric polynomial of order less than $p$ together with  Theorem 7.28 in \cite[Chapter X]{Zygmund}.

Our masked estimator based on $n$ observations $X_1,\ldots,X_n$ of $X\in\RR^p$ is defined as
\[
M\cdot\tilde\Sigma_n=\begin{bmatrix}
\omega_0\tilde\sigma_0 & \omega_1\tilde\sigma_1 & \ldots & \omega_{p-1}\tilde\sigma_{p-1}\\
\omega_1\tilde\sigma_1 & \omega_0\tilde\sigma_0 &  & \vdots\\
\vdots &  & \ddots & \omega_1\tilde\sigma_1\\
\omega_{p-1}\tilde\sigma_{p-1} & \ldots & \omega_1\tilde\sigma_1 & \omega_0\tilde\sigma_0
\end{bmatrix}, 
\]
with
\[
 \tilde\sigma_r=\frac{1}{n}\sum_{i=1}^n\frac{1}{p-r}\sum_{j=1}^{p-r}X_{ij}X_{i(j+r)}, \quad r = 0,\hdots, p-1,
\] 
where $X_{ij}$ is the $j$th entry of the observation $X_i$. Then the corresponding spectral density function is given by
\begin{align}
f_{M\cdot\tilde\Sigma_n}(x)&=\sum_{r=-(p-1)}^{p-1}\omega_{\abs{r}}\tilde\sigma_{\abs{r}}\cos rx=\frac{1}{n}\sum_{i=1}^n\sum_{r=-(p-1)}^{p-1}\frac{\omega_{\abs{r}}}{p-\abs{r}}\sum_{j=1}^{p-\abs{r}}X_{ij}X_{i(j+\abs{r})}\cos rx\notag\\
&=\frac{1}{n}\sum_{i=1}^n\sum_{s=1}^p\sum_{t=1}^p\frac{\omega_{\abs{s-t}}}{p-\abs{s-t}}X_{is}X_{it}\cos(s-t)x=\frac{1}{n}\sum_{i=1}^n\abrac{M\cdot V^xX_i,X_i},\label{eq:SpectralDensityQuadraticFormRepresentation}
\end{align}
where 
\[
V^x=[v_{st}^x]_{s,t=1}^p,\quad v_{st}^x=v_{\abs{s-t}}^x=\frac{\cos(s-t)x}{p-\abs{s-t}},\quad x\in[-\pi,\pi].
\]
Let $Z_i^k$ be the mean-zero random variable defined by
\[
Z_i^k=\abrac{M\cdot V^{x_k}X_i, X_i}-\mean\abrac{M\cdot V^{x_k}X_i, X_i},\quad i=1,\dots,n, \; k=1,\hdots,p.
\]
Then (\ref{eq:SpectralNormMaxSpectralDensity}) and  (\ref{eq:SpectralDensityQuadraticFormRepresentation}) together with the notation above provide the following bound 
\begin{equation}\label{eq:Discretization}
\PP\brac{\norm{f_{M\cdot\tilde\Sigma_n}-f_{M\cdot\Sigma}}_{\infty}\geq t}\leq \PP\brac{\max\limits_{1\leq k\leq 4p}\abs{\frac{1}{n}\sum_{i=1}^nZ_i^k}\geq \frac{t}{2}}.
\end{equation}
By the generalized Hanson-Wright inequality of Theorem \ref{th:HansonWright}, for each $i=1,\hdots,n$ and $k=1,\hdots,4p$, 
\[
\PP\brac{\abs{Z_i^k}\geq t}\leq 2\exp\brac{-\frac{1}{C}\min\brac{\frac{t^2}{2K^4\norm{M\cdot V^{x_k}}_F^2},\frac{t}{K^2\norm{M\cdot V^{x_k}}}}},
\]
which by integration implies that for every integer $q\geq 1$,
\begin{align}
\mean\abs{Z_i^k}^{2q}&\leq 2q\brac{2CK^4\norm{M\cdot V^{x_k}}_F^2}^{q}\Gamma(q)+4q\brac{CK^2\norm{M\cdot V^{x_k}}}^{2q}\Gamma(2q)\notag\\
&\leq q!\brac{4CK^4\norm{M\cdot V^{x_k}}_F^2}^q +(2q)!\brac{2CK^2\norm{M\cdot V^{x_k}}}^{2q}. \label{eq:SubGammaMoments}
\end{align}
According to 
Theorem~\ref{thm:subgamma} it follows that $Z_i^k$ is a sub-gamma random variable with variance factor 
\[
\nu = 16 K^4\brac{C \norm{M\cdot V^{x_k}}_F^2+ C^2\norm{M\cdot V^{x_k}}^2}
\]
and scale parameter 
\[
c = 2CK^2\norm{M\cdot V^{x_k}}.
\]
Hence, by \eqref{def:subgamma} and independence, for all $0 < \lambda < 1/c$,
\begin{equation}\label{eq:SumOfSubGammaRV}
\mean \exp\brac{\lambda \sum_{i=1}^n Z_i^k}=\prod_{i=1}^n\mean\exp\brac{\lambda Z_i^k}\leq\exp\brac{\frac{\lambda^2n\nu}{2(1-c\lambda)}},
\end{equation}
and similarly for the $Z_i^k$ replaced by $- Z_i^k$.
This means that $\displaystyle\sum_{i=1}^n Z_i^k$ is a sub-gamma random variable with variance factor $\nu n$ and scale parameter $c$. 
By \eqref{subgamma:tail} 
this implies that for every $t>0$,
\[
\PP\brac{\abs{\sum_{i=1}^nZ_i^k}>\sqrt{2\nu n t}+ct}\leq 2e^{-t}.
\]
Taking into account that the spectral norm of a matrix is bounded from above by its Frobenius norm we obtain
\[
\PP\brac{\abs{\frac{1}{n}\sum_{i=1}^n Z_i^k}\geq C_1K^2\norm{M\cdot V^{x_k}}_F\sqrt\frac{t}{n}+\frac{CK^2\norm{M\cdot V^{x_k}}t}{n}}\leq 2e^{-t} \quad \mbox{ for all } t > 0,
\]
where $C_1 = 4\sqrt{2C + 2 C^2}$.
A direct calculation of $\norm{M\cdot V^{x_k}}_F$ yields
\[
\norm{M\cdot V^{x_k}}_F\leq \brac{2\sum_{\ell=0}^{p-1}\frac{\omega_{\ell}^2}{p-\ell}}^{1/2}= 2 \|\omega\|_{2,*}.
\]
By the Gershgorin disc theorem \cite[Chapter 6]{HornJohnson}, $\norm{M\cdot V^{x_k}}$ is bounded by
\[
\norm{M\cdot V^{x_k}}\leq 2\sum_{\ell=0}^{p-1}\frac{\omega_{\ell}}{p-\ell}= 2 \|\omega\|_{1,*}.
\]
Applying the union bound to (\ref{eq:Discretization}) results in
\begin{equation}\label{prob:SpectralDensityBound}
\PP\brac{\norm{f_{M\cdot\tilde\Sigma_n}-f_{M\cdot\Sigma}}_{\infty}\geq C_2 K^2\brac{\|\omega\|_{2,*} \sqrt\frac{t}{n}+\frac{\|\omega\|_{1,*} t}{n}}}\leq 8pe^{-t}
\end{equation}
for some $C_2$ only depending on $C$. Due to (\ref{eq:SpectralNormMaxSpectralDensity}) the error of approximating $M\cdot \Sigma$ by $M\cdot\tilde\Sigma_n$ is bounded by
\[
\PP\brac{\norm{M\cdot\tilde\Sigma_n-M\cdot\Sigma}\geq C_2 K^2\brac{\|\omega\|_{2,*} \sqrt\frac{t}{n}+\frac{\|\omega\|_{1,*} t}{n}}}\leq 8pe^{-t}.
\]
Integration yields
\begin{align}
&\mean\norm{f_{M\cdot\tilde\Sigma_n}-f_{M\cdot\Sigma}}_{\infty}\leq C_3 K^2\brac{\|\omega\|_{2,*}\sqrt{\frac{\log(p)}{n}}+\|\omega\|_{1,*}\frac{\log(p)}{n}},\label{eq:MeanSpectralDensityBound}\\
&\mean\norm{M\cdot\tilde\Sigma_n-M\cdot\Sigma}\leq C_3 K^2\brac{\|\omega\|_{2,*}\sqrt{\frac{\log(p)}{n}}+\|\omega\|_{1,*}\frac{\log(p)}{n}}\notag.
\end{align}
This concludes the proof. 
\end{proof}

\begin{proof}[Proof of Corollary \ref{cor:Gaussian}]
The Gaussian distribution satisfies the c.c.p with the constant $K^2=2\norm{\Sigma}$. Since the entries of either the banding or tapering mask $M$ are bounded from above by $1$ and $m\leq\frac{p}{2}$, we obtain
\[
\begin{aligned}
&\|\omega\|_{2,*} \leq \brac{2\sum_{\ell=0}^{m}\frac{1}{p-\ell}}^{1/2}\leq\brac{\frac{2(m+1)}{p-m}}^{1/2}\leq\brac{\frac{4(m+1)}{p}}^{1/2},\\
&\|\omega\|_{1,*}\leq 2\sum_{\ell=0}^{m}\frac{1}{p-\ell}\leq \frac{2(m+1)}{p-m}\leq\frac{4(m+1)}{p}.
\end{aligned}
\]
Theorem \ref{th:MainResult} yields the claim. 
\end{proof}

\begin{proof}[Proof of Theorem \ref{th:Perfomance OfPositiveEstimator}]
By the triangle inequality,
\begin{equation}\label{eq:PosEstTriangleIneq}
\norm{\Sigma^*-\Sigma}\leq\norm{\Sigma^*-M\cdot\Sigma}+\norm{M\cdot\Sigma-\Sigma}.
\end{equation}
We bound the first term by expanding both matrices to a circulant matrix and taking into account expression (\ref{eq:EigenValuesOfCirculantMatrix}) for its eigenvalues,  
\begin{align}
\norm{\Sigma^*-M\cdot\Sigma}&\leq\norm{\Sigma^*_{\circm}-(M\cdot\Sigma)_{\circm}}=\underset{\abs j\leq p-1}\max\abs{f^*\brac{\frac{2\pi j}{2p-1}}-f_{M\cdot\Sigma}\brac{\frac{2\pi j}{2p-1}}}\notag\\
&\leq\norm{f^*-f_{M\cdot\Sigma}}_{\infty}\leq\norm{f^*-f}_{\infty}+\norm{f-f_{M\cdot\Sigma}}_{\infty}\label{eq:PosEstDifferencePosEstAndMaskedCovariance}.
\end{align}
Since $f$ is non-negative and $f^*$ is the positive part of $f_{M\cdot\tilde\Sigma_n}$, 
\begin{equation}\label{eq:PosEstSpectralDensityPositivePart}
\norm{f^*-f}_{\infty}\leq\norm{f_{M\cdot\tilde\Sigma_n}-f}_{\infty}\leq \norm{f_{M\cdot\tilde\Sigma_n}-f_{M\cdot\Sigma}}_{\infty}+\norm{f_{M\cdot\Sigma}-f}_{\infty}.
\end{equation}
Estimating the second term of (\ref{eq:PosEstTriangleIneq}) by the $L^{\infty}$ norm of the corresponding spectral density function together with (\ref{eq:PosEstDifferencePosEstAndMaskedCovariance}) and  (\ref{eq:PosEstSpectralDensityPositivePart}) leads to
\[
\norm{\Sigma^*-\Sigma}\leq \norm{f_{M\cdot\tilde\Sigma_n}-f_{M\cdot\Sigma}}_{\infty}+3 \norm{f-f_{M\cdot\Sigma}}_{\infty}.
\] 
Taking expectations and applying estimate (\ref{eq:MeanSpectralDensityBound}) shows the claimed estimate for the expectation of the approximation error, while a combination with \eqref{prob:SpectralDensityBound} proves the probability bound.
\end{proof}

\bibliographystyle{plain}
\bibliography{ToeplitzRef}

\begin{thebibliography}{10}

\bibitem{AdamczakLogSobolev}
R.~Adamczak.
\newblock Logarithmic {S}obolev inequalities and concentration of measure for
  convex functions and polynomial chaoses.
\newblock {\em Bull. Pol. Acad. Sci. Math.}, 53(2):221--238, 2005.

\bibitem{Adamczak}
R.~Adamczak.
\newblock A note on the {H}anson-{W}right inequality for random vectors with
  dependencies.
\newblock {\em Electron. Commun. Probab.}, 20(72):1--13, 2015.

\bibitem{BickelLevinaThresholding}
P.~Bickel and E.~Levina.
\newblock Covariance regularization by thresholding.
\newblock {\em Ann. Stat.}, 36(6):2577--2604, 2008.

\bibitem{BickelLevina}
P.~Bickel and E.~Levina.
\newblock Regularized estimation of large covariance matrices.
\newblock {\em Ann. Stat.}, 36(1):199--227, 2008.

\bibitem{BoucheronLugosiMassart}
S.~Boucheron, G.~Lugosi, and P.~Massart.
\newblock {\em Concentration inequalities. {A} nonasymptotic theory of
  independence}.
\newblock Oxford University Press, Oxford, 2013.

\bibitem{CaiRenZhou}
{T.} Cai, {Z.} Ren, and {H.} Zhou.
\newblock Optimal rates of convergence for estimating {T}oeplitz covariance
  matrices.
\newblock {\em Probab. Theory Relat. Fields}, 156(1-2):101--143, 2013.

\bibitem{CaiZhou}
T.~Cai and H.~Zhou.
\newblock Minimax estimation of large covariance matrices under
  {$\ell_1$}-norm.
\newblock {\em Stat. Sin.}, 22(4):1319--1349, 2012.

\bibitem{chaudhari}
S.~Chaudhari, V.~Koivunen, and H.~V. Poor.
\newblock Autocorrelation-based decentralized sequential detection of {OFDM}
  signals in cognitive radios.
\newblock {\em IEEE Trans. Signal Process.}, 57(7):2690--2700, 2009.

\bibitem{ChenGittensTropp}
R.~Chen, A.~Gittens, and J.~Tropp.
\newblock The masked sample covariance estimator: an analysis using matrix
  concentration inequalities.
\newblock {\em Inf. Inference}, 1(1):2--20, 2012.

\bibitem{ElKaroui}
N.~El~Karoui.
\newblock Operator norm consistent estimation of large-dimensional sparse
  covariance matrices.
\newblock {\em Ann. Stat.}, 36(6):2717--2756, 2008.

\bibitem{FurrerBengtsson}
R.~Furrer and T.~Bengtsson.
\newblock Estimation of high-dimensional prior and posterior covariance
  matrices in {K}alman filter variants.
\newblock {\em J. Multivariate Anal.}, 98(2):227--255, 2007.

\bibitem{GrenanderSzegoe}
{U.} Grenander and {G.} Szeg{\"o}.
\newblock {\em Toeplitz forms and their applications}.
\newblock California Monographs in Mathematical Sciences. University of
  California Press, Berkeley-Los Angeles, 1958.

\bibitem{HornJohnson}
R.~Horn and C.~Johnson.
\newblock {\em Matrix analysis}.
\newblock Cambridge University Press, Cambridge, 2013.

\bibitem{Ledoux}
M.~Ledoux.
\newblock {\em The concentration of measure phenomenon}, volume~89 of {\em
  Mathematical Surveys and Monographs}.
\newblock American Mathematical Society, Providence, RI, 2001.

\bibitem{LedouxTalagrand}
M.~Ledoux and M.~Talagrand.
\newblock {\em Probability in {B}anach spaces. {I}soperimetry and processes}.
\newblock Springer-Verlag, Berlin, 2011.

\bibitem{LevinaVershynin}
E.~Levina and R.~Vershynin.
\newblock Partial estimation of covariance matrices.
\newblock {\em Probab. Theory Relat. Fields}, 153(3-4):405--419, 2012.

\bibitem{MeckesSzarek}
M.~Meckes and S.~Szarek.
\newblock Concentration for noncommutative polynomials in random matrices.
\newblock {\em Proc. Amer. Math. Soc.}, 140(5):1803--1813, 2012.

\bibitem{Paulin:2014fj}
D.~Paulin.
\newblock {The convex distance inequality for dependent random variables, with
  applications to the stochastic travelling salesman and other problems}.
\newblock {\em Electron. J. Probab.}, 19(0):1--34, 2014.

\bibitem{rauhutward}
H.~Rauhut and R.~Ward.
\newblock Interpolation via weighted $\ell_1$-minimization.
\newblock {\em Appl. Computat. Harmonic Anal.}, 40(2):321--351, 2016.

\bibitem{Samson:2000bf}
P.~Samson.
\newblock {Concentration of measure inequalities for Markov chains and
  $\Phi$-mixing processes}.
\newblock {\em Ann. Prob.}, 28(1):416--461, 2000.

\bibitem{Talagrand:1988ki}
M.~Talagrand.
\newblock {An isoperimetric theorem on the cube and the Kintchine-Kahane
  inequalities}.
\newblock {\em Proc. Am. Math. Soc.}, 104(3):905--909, 1988.

\bibitem{Talagrand:1995tn}
M.~Talagrand.
\newblock {Concentration of measure and isoperimetric inequalities in product
  spaces}.
\newblock {\em Inst. Hautes {\'E}tudes Sci. Publ. Math.}, (81):73--205, 1995.

\bibitem{Vershynin}
{R.} Vershynin.
\newblock Introduction to the non-asymptotic analysis of random matrices.
\newblock In {\em Compressed sensing, {T}heory and {A}pplications}, pages
  210--268. Cambridge Univ. Press, Cambridge, 2012.

\bibitem{VuWang}
V.~Vu and K.~Wang.
\newblock Random weighted projections, random quadratic forms and random
  eigenvectors.
\newblock {\em Random Struct. Algor.}, 47(4):792--821, 2015.

\bibitem{ZengLiang}
Y.~Zeng and Y.~C. Liang.
\newblock Covariance based signal detections for cognitive radio.
\newblock In {\em 2007 2nd IEEE International Symposium on New Frontiers in
  Dynamic Spectrum Access Networks}, pages 202--207, April 2007.

\bibitem{Zygmund}
A.~Zygmund.
\newblock {\em Trigonometric series}.
\newblock Cambridge Mathematical Library. Cambridge University Press,
  Cambridge, 2002.

\end{thebibliography}


\end{document}